\documentclass[prb,twocolumn,showpacs,showkeys]{revtex4}
\usepackage{amsmath,amsgen,amstext,amsbsy,amsopn,amsthm,amssymb,epsf}
\usepackage{amssymb,graphics,graphicx,wrapfig}

\usepackage{amssymb,epsf}

\def\be{\begin{equation}}
\def\ee{\end{equation}}
\def\ba{\begin{align}}
\def\ea{\end{align}}
\def\bsplit{\begin{split}}
\def\esplit{\end{split}}
\def\bm{\begin{multline}}
\def\eem{\end{mutline}}
\def\bfig{\begin{figure}[htb]}
\def\efig{\end{figure}}

\newtheorem{theorem}{Theorem}[section]
\newtheorem{proposition}[theorem]{Proposition}
\newtheorem{lemma}[theorem]{Lemma}

\newcommand{\nn}{\nonumber}

\renewcommand{\leq}{\;\leqslant\;}
\renewcommand{\geq}{\;\geqslant\;}
\newcommand{\dd}{{\rm d}}
\newcommand{\e}[1]{\,{\rm e}^{#1}\,}

\def\Tr{{\operatorname{Tr\,}}}

\newcommand{\upchi}{\raise 2pt \hbox{$\chi$}}

\newcommand{\caS}{{\mathcal S}}

\newcommand{\bbR}{{\mathbb R}}

\newcommand{\bsomega}{{\boldsymbol\omega}}
\newcommand{\bfomega}{{\boldsymbol\omega}}

\begin{document}


\title[Upper Bound on the Critical Temperature of a Bose gas]{Rigorous
  Upper Bound on the Critical Temperature of Dilute Bose Gases}

\author{Robert Seiringer}
\affiliation{Department of Physics, Princeton University, Princeton, NJ 08544, USA}
\email{rseiring@princeton.edu}

\author{Daniel Ueltschi}
\affiliation{Department of Mathematics, University of Warwick, Coventry, CV4 7AL, England}
\email{daniel@ueltschi.org}

\begin{abstract}
  We prove exponential decay of the off-diagonal correlation function
  in the two-dimensional homogeneous Bose gas when $a^2\rho$ is small
  and the temperature $T$ satisfies
\[
T >\frac{ 4\pi \rho }{\ln |\ln(a^2\rho)|}\,.
\]
Here, $a$ is the scattering length of the repulsive interaction
potential and $\rho$ is the density. To leading order in $a^2\rho$,
this bound agrees with the expected critical temperature for
superfluidity. In the three-dimensional Bose gas, exponential decay is
proved when
\[
\frac{T - T_{\rm c}^{(0)}}{T_{\rm c}^{(0)}} > 5 \sqrt{a \rho^{1/3}}\,,
\]
where $T_{\rm c}^{(0)}$ is the critical temperature of the ideal
gas. While this condition is not expected to be sharp, it gives a
rigorous upper bound on the critical temperature for Bose-Einstein
condensation.  \vspace{1mm}
\end{abstract}


\maketitle

\section{Introduction}
\label{sec intro}

Quantum many-body effects due to particle interactions and quantum
statistics make the Bose gas a fascinating system and a challenge to
theoretical physics.  It is increasingly relevant to experimental
physics, especially after the first realization of Bose-Einstein
condensation in cold atomic gases. \cite{AEMWC,DMAVDKW} It displays a
stunning physical phenomenon: superfluidity. Several mechanisms that
are present in the Bose gas also play a r\^ole in interacting
electronic systems and in quantum optics.

Both the two-dimensional and the three-dimensional gas have physical
relevance, and they behave rather differently. We consider them
separately here. Throughout the paper, we shall assume that units are
chosen in such a way that $\hbar=2m=k_{\rm B}=1$, where $m$ is the particle
mass.

\subsection{The two-dimensional Bose gas}

There is no Bose-Einstein condensation in the two-dimensional Bose gas
at positive temperature, as was proved by Hohenberg more than forty
years ago.\cite{Hoh} In contrast to higher dimensions, the
ideal Bose gas offers no intriguing features in two dimensions. But the
interacting gas is expected to display a Kosterlitz-Thouless type
transition from a normal fluid to a superfluid, where the decay of
off-diagonal correlations goes from exponential to power law.  The
critical temperature $T_{\rm c}$ depends on the scattering length $a$
of the interaction potential, which we consider to be repulsive. For
dilute gases, i.e.\ when $a^2 \rho \ll 1$, Popov \cite{Pop} performed
diagrammatic expansions in a functional integral approach, finding
that \be
\label{2D conjecture}
T_{\rm c} \approx \frac{ 4\pi \rho }{\ln |\ln(a^2\rho)|}\,.  
\ee 
This
formula was confirmed by Fisher and Hohenberg \cite{FH} using
Bogoliubov's theory, and by Pilati et al. \cite{PGP} using Monte-Carlo
simulations. No rigorous proof is available to this date, however.

In this article we prove in a mathematically rigorous fashion that
there is exponential decay of the off-diagonal correlation function
when the temperature satisfies \be
\label{ineq no BEC 2d}
T \geq \frac {4\pi \rho}{\ln | \ln(a^2\rho)|}\left( 1 + O \left(
    \frac{\ln\ln|\ln(a^2\rho)|}{\ln |\ln(a^2\rho)|} \right) \right)\ee
for small $a^2\rho$.  Thus we prove that $T_{\rm c}$ cannot be bigger
than the conjectured value \eqref{2D conjecture}, to leading order in $a^2
\rho$. The main novel ingredient in our proof is a rigorous
bound on the grand-canonical density of the interacting Bose
gas. This is explained in the next section.

\subsection{The three-dimensional Bose gas}

A three-dimensional Bose gas is interesting even in the absence of
particle interactions. Bose-Einstein condensation takes place at the
critical temperature $T_{\rm c}^{(0)} = 4\pi
(\rho/\zeta(\frac32))^{2/3}$ (where $\zeta(\frac32) \approx 2.612$, with
$\zeta$ the Riemann zeta function). The effects of particle
interactions on the critical temperature have been studied by many
authors. A consensus has been reached in recent years but it is
tenuous; we give a survey of the main results, both for historical
perspective and in order to gain a sense of the solidity of the
consensus. Let $\Delta T_{\rm c} = T_{\rm c} - T_{\rm c}^{(0)}$ denote
the change of the critical temperature.

\begin{itemize}
\item[1953] Feynman \cite{Fey} argued that interactions increase the effective mass of the particles and hence decrease $T_{\rm c}$, i.e,  $\Delta T_{\rm c} < 0$.
\item[1958] Lee and Yang \cite{LYang} predict that the change of critical temperature is linear in the scattering length, namely
\[
\Delta T_{\rm c} / T_{\rm c}^{(0)} \approx c\, a \rho^{1/3}.
\]
No information on the constant $c$ is provided, not even its sign.
\item[1960] Glassgold, Kaufman, and Watson \cite{GKW} find that the critical temperature increases as $\Delta T_{\rm c}/T_{\rm c}^{(0)} \approx C (a \rho^{1/3})^{1/2}$ with $C>0$.
\item[1964] Huang \cite{Hua1} gives an argument suggesting that $\Delta T_{\rm c}/T_{\rm c}^{(0)} \approx C (a \rho^{1/3})^{3/2}$ with $C>0$.
\item[1971] A Hartree-Fock computation shows that $\Delta T_{\rm c} < 0$ (Fetter and Walecka \cite{FW}).
\item[1982] A loop expansion of the quantum field representation gives $\Delta T_{\rm c}/T_{\rm c}^{(0)} \approx -3.5 (a \rho^{1/3})^{1/2}$ (Toyoda \cite{Toy}).
\item[1992] By studying the evolution of the interacting Bose gas, Stoof \cite{Sto} finds that the change of critical temperature is linear in the scattering length with $c = 16 \pi / 3 \zeta(3/2)^{4/3} = 4.66$.
\item[1996] A diagrammatic expansion in the renormalization group yields $\Delta T_{\rm c} > 0$ (Bijlsma and Stoof \cite{BS}).
\item[1997] A path integral Monte-Carlo simulation yields $c=0.34\pm 0.06$ (Gr\"uter, Ceperley, and Lalo\"e \cite{GCL}).
\item[1999] A virial expansion leads to $c=0.7$ (Holzmann, Gr\"uter, and Lalo\"e \cite{HGL}). Another virial expansion leads Huang \cite{Hua2} to conclude that $\Delta T_{\rm c}/T_{\rm c}^{(0)} \approx 3.5 (a \rho^{1/3})^{1/2}$. Interchanging the limit $a\to0$ with the thermodynamic limit, and using Monte-Carlo simulations, Holzmann and Krauth \cite{HK} find $c=2.3\pm 0.25$. The dilute Bose gas can be mapped onto a classical field lattice model (Baym et.\ al.\ \cite{BBHLV}); a self-consistent approach then yields $c=2.9$.
\item[2000] An experimental realization by Reppy et.\ al.\ \cite{RCHCHZ} yields $c=5.1\pm 0.9$. It was later pointed out that the estimation of the scattering length between particles was not correct, however. \cite{KPS}
\item[2001] Arnold and Moore, \cite{AM} and Kashurnikov, Prokof'ev, and Svistunov \cite{KPS} performed numerical simulations on the equivalent classical field model; \cite{BBHLV} the former get $c=1.32\pm 0.02$ and the latter get $c=1.29\pm 0.05$.
\item[2003] A variational perturbation theory performed by Kleinert \cite{Kle} yields $c=1.14\pm 0.11$.
\item[2004] By studying the classical field model with variational perturbations, Kastening \cite{Kas} finds $c=1.27\pm 0.11$.  A path integral Monte-Carlo simulation by Nho and Landau \cite{NL} yields $c=1.32\pm 0.14$.
\end{itemize}
The last articles essentially agree with one another, and also with more recent articles. \cite{PGP}
The case for a linear
correction with constant $c \approx 1.3$ is made rather convincingly;
it is not beyond reasonable doubt, though. Notice that the constant
$c$ is {\it universal} in the sense that it does not depend on such
special features as the mass of the particles or the details of the
interactions. (The mass enters the scattering length
$a$, however.)

The question of the critical temperature for interacting Bose gases
is reviewed in Baym et.\ al. \cite{BBHLV2} and in Blaizot. \cite{Bla} A
comprehensive survey on many aspects of bosonic systems has been
written by Bloch, Dalibard, and Zwerger. \cite{BDZ}
This question is also mentioned in additional articles dealing with certain perturbation methods.
The value of $c$ is assumed to be known and its calculation serves to test the method.
Some of these references can be found in Blaizot. \cite{Bla}

In this article we give a partial rigorous justification of the
results in the literature by proving that off-diagonal
correlations decay exponentially when 
\be\label{ineq no BEC} \frac{T -
  T_{\rm c}^{(0)}}{T_{\rm c}^{(0)}} \geq 5.09 \sqrt{a \rho^{1/3}} \,
\left( 1 + O\big(\sqrt{a\rho^{1/3}}\big)\right)\,.  
\ee 
In particular,
there is no Bose-Einstein condensation when (\ref{ineq no BEC}) is
satisfied. This rigorous result is not sharp enough to disprove any of
the previous claims that have been just reviewed,
although it gets close to Huang's 1999 result. As in the
two-dimensional case, the proof is based on bounds of the
grand-canonical density for the interacting gas.

\subsection{Outline of this article}

In the next section, we shall explain how the exponential decay of
correlations can be deduced from appropriate lower bounds on the
particle density in the grand-canonical ensemble. These bounds will be
proved in the remaining sections. In Section~\ref{sec density
  bounds}, we shall state our main result, Theorem~\ref{thm density
  bounds}, and we shall explain the precise assumptions on the
interparticle interactions under which it holds. Our main tool is a
path integral representation which is explained in detail in
Section~\ref{sec FK}. Finally, in Section~\ref{sec scattering} we
investigate certain integrals of the difference between the heat kernel
of the Laplacian with and without potential, and obtain bounds that are
needed to complete the proof of Theorem~\ref{thm density
  bounds}.

\section{Decay of correlations}

We consider the grand-canonical ensemble at chemical potential $\mu$
and we denote the fugacity by $z = \e{\beta\mu}$. Let $\gamma(x,y) =
\langle a^\dagger(x) a(y) \rangle$ denote the reduced one-particle
density matrix of the interacting system, and $\gamma^{(0)}$ the one
of the ideal gas. An important fact is that, when the interactions are
repulsive, we have \be\label{imp} \gamma(x,y) \leq \gamma^{(0)}(x,y)
\ee for any $0<z<1$.  See Bratteli-Robinson, \cite{BR}
Theorem~6.3.17. In $d$ spatial dimensions,
\[
\gamma^{(0)}(x,y) = \sum_{n\geq1} \frac{z^n}{(4\pi\beta n)^{d/2}}
\e{-\frac{|x-y|^2}{4\beta n}} 
\]
which behaves like $\exp ( - \sqrt{-\beta^{-1}\ln z} \,
|x-y|)$ for large $|x-y|$. That is, off-diagonal correlations decay
exponentially fast when $z<1$. In particular, the critical fugacity satisfies $z_{\rm c}\geq 1$. 

Next, let $\rho(z)$ denote the grand-canonical density of the interacting system (it depends on $\beta$ as well, although the notation does not show it explicitly), and let
\be\label{rhozero}
\rho^{(0)}(z) = (4\pi\beta)^{-d/2} g_{d/2}(z)
\ee
the density of the ideal system. Here, the function $g_{d/2}$ is defined by
\be
\label{def g}
g_r(z) = \sum_{n\geq1} \frac{z^n}{n^r}\,.
\ee
The density $\rho(z)$ is increasing in $z$. Then a sufficient condition for the exponential decay of correlations is that, for some $z<1$,
\be
\label{the criterion}
\rho < \rho(z)\,.  \ee The obvious problem with this condition is that
the density $\rho(z)$ for the interacting system is not given by an
explicit function. Our way out is to obtain bounds for $\rho(z)$ (see
Theorem \ref{thm density bounds} below) and to use them with $z <1$
suitably chosen.

\subsection{Two dimensions}

We now explain the proof of exponential decay of correlations under the condition \eqref{ineq no BEC 2d} for $d=2$. We show below (see Theorem \ref{thm density bounds} and the following remarks) that the density satisfies the lower bound 
\be\label{rho2}
\rho(z) \geq \rho^{(0)}(z) - \frac C{4\pi\beta} \frac{|\ln(1-z)|}{1-z} \frac1{|\ln(a^2/\beta)|}\,,
\ee
for some constant $C>0$ and for $a\beta^{-1/2}$ small enough. Here, $a$ denotes the two-dimensional scattering length, which can be defined similarly to the three-dimensional case
via the solution of the zero-energy scattering equation.\cite{LY2d,LSSY}

In two dimensions, $\rho^{(0)}(z) = - (4\pi\beta)^{-1} \ln(1-z)$. For the choice $z=z_0$ with 
\[
z_0 = 1 - \frac{\ln |\ln(a^2/\beta)|}{|\ln(a^2/\beta)|}\,,
\]
the criterion \eqref{the criterion} is fulfilled when
\[
\rho \leq \frac{\ln |\ln(a^2/\beta)|}{4\pi\beta} \Bigl( 1 - O\Bigl( \frac{\ln\ln |\ln(a^2/\beta)|}{\ln |\ln(a^2/\beta)|} \Bigr) \Bigr).
\]
Since $\beta=1/T$, one can check that this is equivalent to the condition \eqref{ineq no BEC 2d}. 

The situation is illustrated in Fig.\ \ref{fig2D} with qualitative
graphs of $\rho^{(0)}(z)$ and $\rho(z)$. The critical fugacity $z_{\rm
  c}$ is known to be larger than 1. Our density bound holds for $z <
1$, and this yields the lower bound $\rho_0$ for the critical
density. It turns out to be equal to the conjectured
critical density (determined by Eq.\ \eqref{2D conjecture}) to leading order in the
small parameter $a\beta^{-1/2}$.

\bfig
\centerline{\includegraphics[width=70mm]{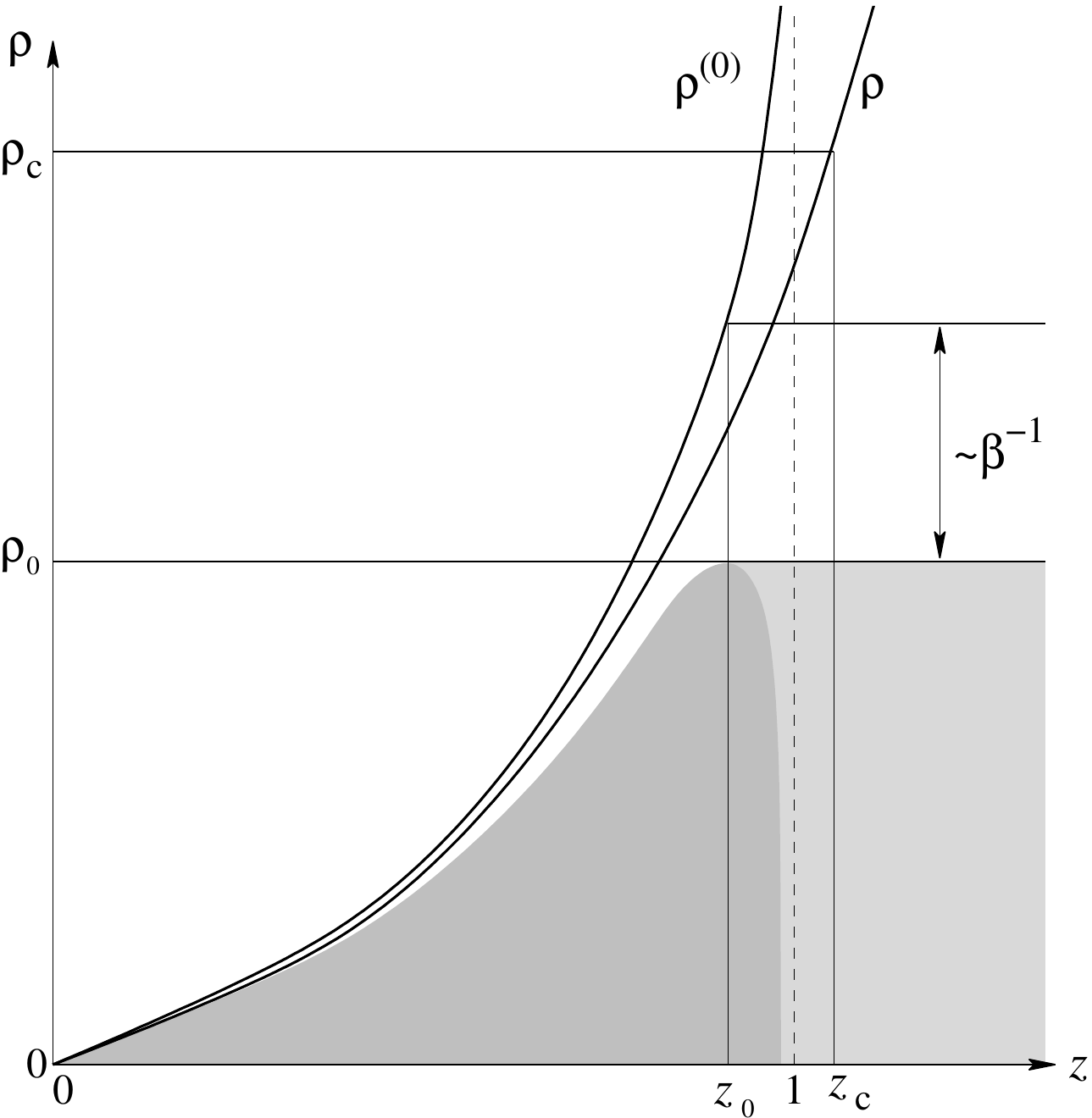}}
\caption{Qualitative graphs of the grand-canonical density for $d=2$. The shaded area represents our lower bound for the interacting density --- the darker area is the function defined in Eq.\ \eqref{rho2} and it extends to the lighter area by monotonicity of the grand-canonical density. Our lower bound $\rho_0$ for the critical density is obtained by choosing $z_0 = 1 - \frac{\ln |\ln(a^2/\beta)|}{|\ln(a^2/\beta)|}$.}
\label{fig2D}
\efig

\subsection{Three dimensions}

We shall prove exponential decay under the 
condition \eqref{ineq no BEC}, where the constant $5.09$ is really
\[
A = \frac{2^{7/2} \pi^{1/2}}{3 \, \zeta(3/2)^{7/6}}\sqrt{2^{3/2} + \zeta(3/2)} \approx 5.09.
\]
It is more convenient to consider the change in the critical
density rather than in the temperature. Inequality \eqref{ineq no
 BEC} is equivalent to
\be
\label{the equivalent condition}
\frac{\rho - \rho_{\rm c}^{(0)}}{\rho_{\rm c}^{(0)}} \leq -A' \sqrt{a
  \beta^{-1/2}} \, \left( 1 + O\big(\sqrt{a\beta^{-1/2}}\big)\right),
\ee
where $\rho_{\rm c}^{(0)} = \rho^{(0)}(1)$ is the critical density of the ideal Bose gas at temperature $T$, and 
where the constants $A$ and $A'$ are related by 
\[
 A' = \frac 32 \frac{\zeta(3/2)^{1/6}}{(4\pi)^{1/4}} A \approx 4.75 \,.
 \]
 
We show below that the lower bound 
\be\label{rho3}
\rho(z) \geq \rho^{(0)}(z)
- \frac a{(2\pi\beta)^2} \Bigl[ \bigl( 2^{3/2} + \zeta(\tfrac32) \bigr) \sqrt{\frac\pi{-\ln z}} + C \Bigr]
\ee
holds for some positive constant $C$ and $a\beta^{-1/2}$ small enough (see Theorem \ref{thm density bounds} and the following remarks).
We use $\dd g_{3/2}/\dd z = z^{-1} g_{1/2}(z)$, as well as 
the bound 
\be\label{sumin}
g_{1/2}(z) \leq \int_0^\infty \frac{z^t}{\sqrt{t}} \dd t = \sqrt{\frac\pi{-\ln z}}
\ee
to obtain
\[
\begin{split}
  \rho^{(0)}(1) - \rho^{(0)}(z) &\leq (4\pi\beta)^{-\frac32} \int_z^1 \sqrt{\frac\pi{-\ln s}} \, \frac{\dd s}s \\
  &= (4\pi)^{-1} \beta^{-\frac32} \sqrt{-\ln z}.
\end{split}
\]
The criterion \eqref{the criterion} is thus fulfilled when
\bm\nn
\rho \leq \rho^{(0)}(1) - (4\pi)^{-1}\beta^{-\frac32} \sqrt{-\ln z} \\
-  \frac a{(2\pi\beta)^2}\Bigl[ \bigl( 2^{3/2} + \zeta(3/2) \bigr) \sqrt{\frac\pi{-\ln z}} + C \Bigr] 
\end{multline}
for some $z<1$. The right side of this expression depends on $z$ only through $w=\sqrt{-\ln z}$. Since the minimum of $Aw + \frac Bw$ over $w>0$ is $2\sqrt{AB}$, we get the condition \eqref{the equivalent condition}. Notice that the optimal choice of $z$ is $z_0 = 1 - \pi^{-1/2} (2^{3/2} + \zeta(3/2)) a\beta^{-1/2}$ to leading order in $a\beta^{-1/2}$.

The three-dimensional situation is illustrated in Fig.\ \ref{fig3D}. The critical fugacity $z_{\rm c}$ is larger than 1 but our density bound holds for $z < 1$. Our lower bound for the critical density, $\rho_0$, is close to the conjectured expression for small $a\beta^{-1/2}$.

\bfig
\centerline{\includegraphics[width=70mm]{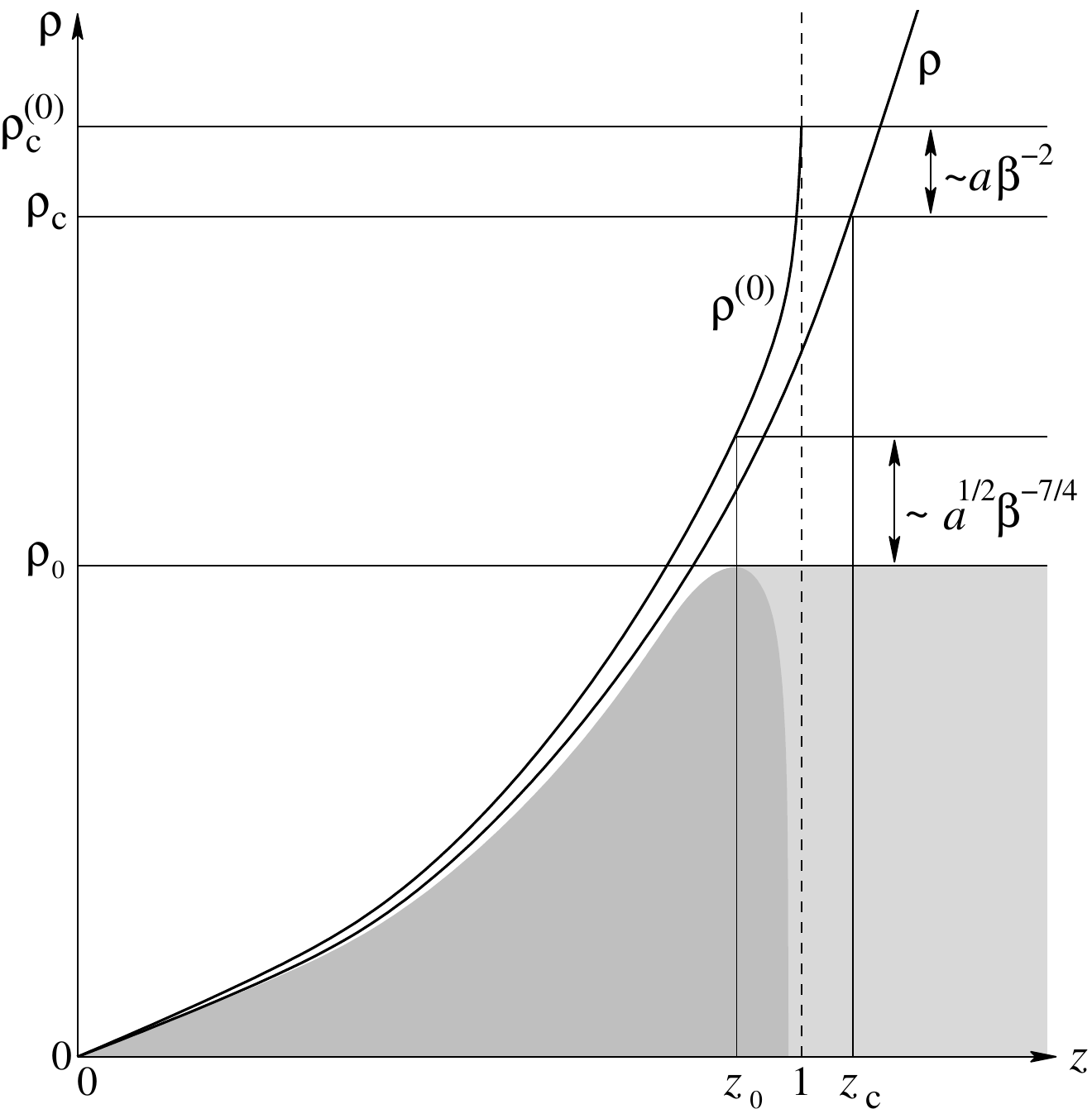}}
\caption{Qualitative graphs of the grand-canonical density for $d=3$. The shaded area represents our lower bound for the interacting density --- the darker area is the function defined in Eq.\ \eqref{rho3} and it extends to the lighter area by monotonicity of the grand-canonical density. Our lower bound $\rho_0$ for the critical density is obtained by choosing $z_0 = 1 - C a\beta^{-1/2}$. The difference between $\rho^{(0)}_{\rm c}$ and $\rho_{\rm c}$ is expected to be of the order $a\beta^{-2}$. }
\label{fig3D}
\efig

\section{Rigorous density bounds}
\label{sec density bounds}

We are left with proving the lower bounds (\ref{rho2}) and
(\ref{rho3}), respectively. These will be an immediate consequence of
Theorem~\ref{thm density bounds} below. In order to state our results
precisely, we shall first give a definition of the model and specify
the assumptions on the interaction potential. This makes it necessary
to adopt a precise mathematical tone from now on. We do so in order to
make the results accessible also to readers with a more mathematical
background.

Let $\Lambda \subset \bbR^d$ be an open and bounded domain.  The state
space for $N$ bosons in $\Lambda$ is the Hilbert space $L^2_{\rm
  sym}(\Lambda^N)$ of square-integrable complex-valued functions that
are symmetric with respect to their arguments. The Hamiltonian is
\be\nn H_{\Lambda,N} = -\sum_{i=1}^N \Delta_i + \sum_{1\leq i<j \leq
  N} U(x_i-x_j), \ee with $\Delta_i$ the Laplacian for the $i$-th
variable, with Dirichlet boundary conditions on the boundary of
$\Lambda$. The repulsive interaction is given by the multiplication
operator $U(x) \geq 0$. We assume that $U$ is radial and has finite
range, i.e., $U(x)=0$ for $|x|>R_0$. No regularity is assumed,
however; we only require that the Hamiltonian defines a self-adjoint
operator on an appropriate domain, and that the Feynman-Kac formula
for the heat kernel applies. In particular, $U$ is allowed to have a
hard core. The scattering length of $U$ is denoted by $a$.

The grand-canonical partition function is
\[
Z\equiv Z(\beta,\Lambda,z) = \sum_{N\geq0} z^N \Tr \e{-\beta H_{\Lambda,N}}.
\]
The thermodynamic pressure is defined by
\[
p(\beta,z) = \frac1{\beta |\Lambda|} \ln Z(\beta,\Lambda,z),
\]
and the density is given by
\be\label{def rho}
\rho(z) = \beta z \frac\partial{\partial z} p(\beta,z).
\ee

We always work in finite volume $\Lambda$. The existence of the
thermodynamic limit for the pressure, density and reduced density
matrix is far from trivial. In particular, the limit for the latter
has only been proved when $z$ is small enough.\cite{BR} This is of no
relevance to the present article, however, since our bounds apply to
all finite domains uniformly in the volume.  The one-particle reduced
density matrix can be written in terms of the integral kernels of the
operators $\e{-\beta H_{\Lambda,N}}$ as \bm\nn
\gamma(x,y) = \frac1{Z} \sum_{N\geq1} N z^N
\int_{\Lambda^{N-1}} \dd x_2 \cdots \dd x_N \\ \times \e{-\beta
  H_{\Lambda,N}}(x,x_2,\dots,x_N;y,x_2,\dots,x_N) \,.
\end{multline}

Relatively few rigorous results on interacting homogeneous Bose gases
are available to this date.  The only proof of occurrence of
Bose-Einstein condensation deals with the hard-core lattice model at
half-filling. \cite{DLS,KLS} Roepstorff \cite{Roe} used Bogoliubov's
inequality to get an upper bound on the condensate density. Several
aspects of Bogoliubov's theory \cite{ZB,LSSY} have been rigorously
justified. \cite{Gin1,LSY,Suto2} A rigorous proof of the leading order
of the ground state energy per particle in the low density limit was
given by Lieb and Yngvason. \cite{LY,LY2d} The next order correction
term was recently studied in a certain scaling limit. \cite{SG} Bounds
of the free energy at positive temperature were given in \cite{Sei}.
Cluster expansions give informations on the phase without
Bose-Einstein condensation, for repulsive or stable potentials.
\cite{Gin,PU} Recently there has been interest in Feynman cycles which
should be related to Bose-Einstein condensation. \cite{Suto1} The
conditions \eqref{ineq no BEC 2d} and \eqref{ineq no BEC} guarantee
the absence of infinite cycles. This follows from the considerations
here, and from the proof that all cycles are finite when the chemical
potential is negative. \cite{Uel}

The following theorem gives bounds on the density $\rho(z)$. Recall the function $g_r$ defined in \eqref{def g}.
Let us define the following small parameter $\widetilde a(\beta)$, which is associated with the scattering length $a$: for $d=2$,
\[
\widetilde a(\beta) = \left( |\ln(a^2/\beta)| - 2 \ln |\ln (a/\sqrt\beta) |\right)^{-1}+ |\ln(a^2/\beta)|^{-2};
\]
and for $d=3$,
\[
\widetilde a(\beta) = a \Bigl( \big[1 - (a/\sqrt\beta)^{1/2} \big]^{-1} + \tfrac 13 (a/\sqrt\beta)^{1/2} \Bigr).
\]

\begin{theorem}
\label{thm density bounds}
Let us assume that $\sqrt\beta\,|\ln(a/\sqrt\beta)|^{-1}>R_0$ when $d=2$, or that $a\sqrt\beta > R_0^2$ when $d=3$. Then we have, for $0<z<1$, 
\be
\label{2in}
\rho(z) \geq \rho^{(0)}(z) - \frac{4 z^2}{(4\pi\beta)^{d-1}} \left( h_d(z)  \widetilde a(\beta)+ 2^{d/2} \widetilde a(\beta/2)\right),
\ee
where
\be\label{defh}
h_d(z) =\left( 2^{\frac d2} + g_{\frac d2}(z)\right) g_{\frac d2-1}(z)  + 2^{\frac d2+1} g_{\frac d2}(z) + g_{\frac d2}(z)^2\,.
\ee
\end{theorem}

Notice that $\rho(z) \leq \rho^{(0)}(z)$; this is an immediate
consequence of \eqref{imp}.  For $d=2$ we believe that for $z$ close
to 1 the lower bound is optimal up to terms of higher order in $\tilde
a(\beta)$, while for $d=3$ the prefactor is not optimal. This is based
on the (yet unproved) assumption that the leading order correction to
the pressure is equal to $-8\pi \widetilde a(\beta) \rho^{(0)}(z)^2$ for
$z<1$.\cite{LY2d,Sei,HuangBook}. Using \eqref{rhozero} and \eqref{def rho}, this suggests that $\rho(z)\approx \rho^{(0)}(z) - 4 \widetilde
a(\beta) (4\pi\beta)^{1-d} g_{d/2}(z) g_{d/2-1}(z)$. If this indeed holds as a
lower bound, one can replace the constant $5.09$ in \eqref{ineq no BEC} by
$3.52$, yielding a bound in agreement with Huang's prediction.\cite{Hua2}

From Theorem~\ref{thm density bounds}, we can easily deduce the bounds \eqref{rho2} and \eqref{rho3}, which we have used in the previous section.
Since $g_0(z) = z/(1-z)$ and $g_1(z)= -\ln(1-z)$, we  see that the function $h_2$ is bounded by
\[
h_2(z) \leq C (1-z)^{-1} |\ln(1-z)|
\]
for some constant $C<\infty$, which implies \eqref{rho2}. To obtain \eqref{rho3}, note that the function $g_{3/2}(z)$ converges to $\zeta(3/2)$ as $z \to 1$. Using the bound 
(\ref{sumin})
we see that  $h_3$ is less than
\[
h_3(z) \leq \bigl( 2^{3/2} + \zeta(\tfrac32) \bigr) \sqrt{\frac\pi{-\ln z}} + 2^{5/2} \zeta(\tfrac32) + \zeta(\tfrac32)^2\,.
\]

We are left with the proof of Theorem~\ref{thm density bounds}.  In
Section \ref{sec FK} we use the Feynman-Kac representation of the Bose
gas to obtain bounds on the density. These bounds are expressed in
terms of integrals of the difference between the heat kernel of the
Laplacian with and without potential. Section \ref{sec scattering}
deals with bounds of these integrals.  It contains a novel variational
principle for integrals over heat kernel differences
(Lemma~\ref{varpr}), which allows to bound these in terms of the
scattering length of the interaction potential. Theorem \ref{thm
  density bounds} then follows directly from Proposition \ref{prop
  technical bounds} and from Lemmas \ref{lem scattering 1} and
\ref{lem scattering 2}.

\section{Feynman-Kac representation of the interacting Bose gas}
\label{sec FK}

From now on we shall work in arbitrary dimension $d\geq 1$. 
Let $W_{x,y}^t$ denote the Wiener measure for the Brownian bridge from $x$ to $y$ in time $t$; the normalization is chosen so that
\be\nn
\int\dd W_{x,y}^t(\omega) = (2\pi t)^{-d/2} \e{-|x-y|^2/2t} \equiv \pi_t(x-y).
\ee
The integral kernel of $\e{2\beta\Delta} - \e{\beta (2\Delta -  U)}$ will be denoted by $K(x,y)$. 
By the Feynman-Kac formula, it can be expressed as
\be
\label{FK kernel}
K(x,y) = \int \bigl( 1 - \e{-\frac14 \int_0^{4\beta} U(\omega(s)) \dd s} \bigr) \, \dd W^{4\beta}_{x,y}(\omega).
\ee
Let us introduce the interaction $\overline U(\omega,\omega')$ between two paths $\omega$ and $\omega'$: $[0,2\beta] \to \bbR^d$. Namely,
\[
\overline U(\omega,\omega') = \tfrac12 \int_0^{2\beta} U \bigl( \omega(s) - \omega'(s) \bigr) \dd s.
\]
The following identity, which will prove useful in the sequel, is obtained by changing to center-of-mass and relative coordinates.

\begin{lemma}
\label{lem diff B bridges}
For any $x,y,x',y'\in\bbR^d$,
\bm
\int\dd W_{x,y}^{2\beta}(\omega) \int\dd W_{x',y'}^{2\beta}(\omega') \bigl( 1 - \e{-\overline U(\omega, \omega')} \bigr) \\
= 2^d \pi_{4\beta}(x-y+x'-y') K(x-x', y-y'). \nn
\end{multline}
\end{lemma}

\begin{proof}
The difference $\omega-\omega'$ of two Brownian bridges is a Brownian bridge with double variance. Precisely, we have
\bm\nn
\int\frac{\dd W_{x,y}^{2\beta}(\omega)}{\pi_{2\beta}(x-y)} \int\frac{\dd W_{x',y'}^{2\beta}(\omega')}{\pi_{2\beta}(x'-y')} \bigl( 1 - \e{-\overline U(\omega, \omega')} \bigr) \\
= \int\frac{\dd W_{x-x',y-y'}^{4\beta}(\omega)}{\pi_{4\beta}(x-x'-y+y')} \Bigl( 1 - \e{-\frac12 \int_0^{2\beta} U(\omega(2s)) \dd s} \Bigr).
\end{multline}
By the parallelogram identity,
\be\nn
\frac{\pi_{2\beta}(x-y) \pi_{2\beta}(x'-y')}{\pi_{4\beta}(x-x'-y+y')} = 2^d \pi_{4\beta}(x-y+x'-y').
\ee
The result then follows from \eqref{FK kernel}.
\end{proof}

We also use the Feynman-Kac formula for the canonical partition function. Namely,
\begin{align}\nn 
&\Tr \e{-\beta H_{\Lambda,N}} = \frac1{N!} \sum_{\pi\in\caS_N} \int_{\Lambda^N} \dd x_1 \cdots \dd x_N  \\  \nn & \qquad \times
\int\dd W_{x_1,x_{\pi(1)}}^{2\beta}(\omega_1) \cdots \int\dd W_{x_N,x_{\pi(N)}}^{2\beta}(\omega_N)  \\  & \qquad \times \Bigl( \prod_{i=1}^N \upchi_\Lambda(\omega_i) \Bigr)
\exp\biggl\{ -\sum_{1\leq i<j\leq N} \overline U(\omega_i, \omega_j) \biggr\}\,.\label{FK rep}
\end{align}
Here, $\caS_N$ is the set of permutations of $N$ elements; $\upchi_\Lambda(\omega)$ is equal to one if $\omega(s) \in \Lambda$ for all $0 \leq s \leq 2\beta$, and it is zero otherwise. Eq.\ \eqref{FK rep} makes sense for general measurable functions $U : \bbR^d \to \bbR \cup \{\infty\}$. In particular, we can consider the case of the hard-core potential of radius $a$. An introduction to the Feynman-Kac formula in the context of bosonic quantum systems can be found in Ginibre's survey. \cite{Gin}

We now rewrite the grand-canonical partition function in terms of winding loops. Let $\Omega_k$ be the set of continuous paths $[0,2\beta k] \to \bbR^d$ that are closed. Its elements are denoted by $\bfomega = (x,k,\omega)$, with $x \in \bbR^d$ the starting point and $k$ the winding number; we have $\omega(0) = \omega(2\beta k) = x$. For $0 \leq \ell \leq k-1$, we also let $\omega_\ell$ denote the $\ell$-th leg of $\bsomega$,
\[
\omega_\ell(s) = \omega(2\beta\ell+s),
\]
with $0 \leq s \leq 2\beta$.
We consider the measure $\mu$ given by
\be\nn
\dd\mu(\bfomega) = \frac{z^k}k \dd x\, \upchi_\Lambda(\omega) \dd W_{x,x}^{2\beta k}(\omega) \e{-V(\bfomega)}.
\ee
Here, $V(\bfomega)$ is a self-interaction term that is defined below in Eq.\ \eqref{def V}. Let $\Omega = \cup_{k\geq1} \Omega_k$; the measure $\mu$ above naturally extends to a measure on $\Omega$. The grand-canonical partition function can then be written as\cite{Gin}
\bm\nn
Z = \sum_{n\geq0} \frac1{n!} \int_{\Omega^n} \dd\mu(\bfomega_1) \cdots \dd\mu(\bfomega_n) \\ \times
\exp\Bigl\{ -\sum_{1\leq i<j\leq n} V(\bfomega_i,\bfomega_j) \Bigr\}.
\end{multline}
The self-interaction $V(\bfomega)$ and the 2-path interaction $V(\bfomega,\bfomega')$ are given by
\be
\label{def V}
\begin{split}
&V(\bfomega) = \sum_{0\leq\ell<m\leq k-1} \overline U(\omega_\ell, \omega_m), \\
&V(\bfomega,\bfomega') = \sum_{\ell=0}^{k-1} \sum_{\ell'=0}^{k'-1} \overline U(\omega_\ell, \omega_{\ell'}').
\end{split}
\ee
We shall denote $V_{ij}\equiv V(\bfomega_i,\bfomega_j)$ for short.
Using (\ref{def rho}) one obtains an expression for the grand-canonical density, namely
\bm
\label{density FK}
\rho(z) = \frac1{|\Lambda| Z} \sum_{n\geq1} \frac1{(n-1)!} 
\int\dd\mu(\bfomega_1) k_1 \\ \times \int\dd\mu(\bfomega_2) \cdots \int\dd\mu(\bfomega_n) \e{-\sum_{i<j} V_{ij}}.
\end{multline}

From the representation (\ref{density FK}) it is easy to see that $\rho(z)\leq \rho^{(0)}(z)$. One uses the positivity of $V_{ij}$ to bound $\sum_{1\leq i<j\leq N} V_{ij} \geq \sum_{2\leq i<j\leq N} V_{ij}$. For fixed $\bfomega_1$, the integration over $\bfomega_j$ with $j\geq 2$ then yields $Z$, and hence 
$$
\rho(z) \leq \frac{1}{|\Lambda|} \int \dd\mu(\bfomega_1) k_1 \leq \rho^{(0)}(z)\,.
$$
The last inequality follows since the self-interaction $V(\bfomega)$ is also positive.

In the following proposition we shall derive a {\it lower} bound on $\rho(z)$. We use the function $h_d$ defined in \eqref{defh}, as well as the integral kernel $K(x,y)$ in (\ref{FK kernel}). 

\begin{proposition}
\label{prop technical bounds}
For $d\geq 1$, $\beta>0$ and $0<z<1$, we have the lower bound
\begin{align}\nn 
\rho(z) \geq &\rho^{(0)}(z) - \frac{2 z^2}{(4\pi\beta)^{d}} \left( h_d(z) \int K(x,y) \dd x \dd y \right. \\
& +\left. \frac 12 (8\pi\beta)^{d/2} \int \left[ K(x,x) + K(x,-x) \right] \dd x\right)
 \nn
\end{align}
 for any bounded (and measurable) $\Lambda \subset \bbR^d$.
\end{proposition}

\begin{proof}
Isolating the interactions between the first path and the others, we can bound $\exp\{-\sum_{1\leq i<j\leq n} V_{ij} \}$ from below as
\begin{align}\nonumber
&\exp\Bigl\{ -\sum_{j=2}^n V_{1j} \Bigr\} \exp\Bigl\{ -\sum_{2\leq k<l\leq n} V_{kl} \Bigr\} \\
&\geq \Bigl[ 1 - \sum_{j=2}^n (1 - \e{-V_{1j}}) \Bigr] \exp\Bigl\{ -\!\!\!\!\! \sum_{2\leq k<l\leq n} V_{kl} \Bigr\}. \label{ineq inter}
\end{align}
We use this lower bound in Eq.~\eqref{density FK}. The first term (the
$1$ in square brackets in (\ref{ineq inter})) is then 
$|\Lambda|^{-1} \int\dd\mu(\bfomega_1) k_1$, since the integration
over $\bfomega_j$ with $j\geq 2$ yields exactly $Z$. For the remaining terms (the sum over $j$), we also use the fact that the potential is repulsive so as to
drop the interactions between $\bfomega_j$ and the other loops in the last term in (\ref{ineq inter}) for a lower
bound. We conclude that
\begin{align}\label{a lower bound}
\rho(z) & \geq \frac1{|\Lambda|}  \int\dd\mu(\bfomega) k \\ & \qquad - \frac1{|\Lambda|} \int\dd\mu(\bfomega_1) k_1 \int\dd\mu(\bfomega_2) (1 - \e{-V_{12}}) \,. \nn
\end{align}

In a similar fashion to \eqref{ineq inter}, we have
\be\nn
\begin{split}
&\e{-V(\bfomega)} \geq 1 - \sum_{0\leq\ell<m\leq k-1} \Bigl( 1 - \e{-\overline U(\omega_\ell, \omega_m)} \Bigr), \\
&\e{-V(\bfomega_1,\bfomega_2)} \geq 1 - \sum_{\ell_1=0}^{k_1-1} \sum_{\ell_2=0}^{k_2-1} \Bigl( 1 - \e{-\overline U(\omega_{1,\ell_1}, \omega_{2,\ell_2})} \Bigr).
\end{split}
\ee
Here, $\omega_{i,\ell}$ denotes the $\ell$-th leg of the path $\bsomega_i$.
We insert these inequalities into \eqref{a lower bound}, and obtain
\be\nn
\rho(z) \geq \rho^{(0)}(z) - A - B,
\ee
with
\be
\begin{split}\nn
&A = \frac1{|\Lambda|} \sum_{k\geq2} z^k \int_\Lambda \dd x \int\dd W_{x,x}^{2\beta k}(\omega) \\
&\hspace{2.5cm} \times \sum_{0\leq\ell<m\leq k-1} \Bigl( 1 - \e{-\overline U(\omega_\ell, \omega_m)} \Bigr)\,, \\
&B = \frac1{|\Lambda|} \int\dd\mu(\bfomega_1) k_1^2 \int\dd\mu(\bfomega_2) k_2 \Bigl( 1 - \e{-\overline U(\omega_{1,1}, \omega_{2,1})} \Bigr).
\end{split}
\ee
We also used $\chi_\Lambda\leq 1$ to drop the restriction that paths stay inside $\Lambda$. 
Notice that only the first legs of $\omega_1$ and $\omega_2$ interact in $B$; this is correct because we multiplied by $k_1 k_2$. 

We decompose the terms in $A$ as $A_1+A_2+A_3$ according to the distance between interacting legs. Namely, the term $k=2$ in $A$ is equal to
\bm\nn
A_1 = \frac{z^2}{|\Lambda|} \int_{\Lambda^2} \dd x_1 \dd x_2 \int\dd W_{x_1,x_2}^{2\beta}(\omega_1) \\ \times \int\dd W_{x_2,x_1}^{2\beta}(\omega_2) 
\Bigl( 1 - \e{-\overline U(\omega_{1,1}, \omega_{2,1}(s)} \Bigr).
\end{multline}
Using Lemma \ref{lem diff B bridges}, we get
\be \nn 
A_1 \leq \frac{z^2}{(2\pi\beta)^{d/2}} \int K(x,-x) \dd x.
\ee
The terms with $k\geq3$ and two consecutive interacting legs are
\bm\nn
A_2 = \frac1{|\Lambda|} \int_{\Lambda^3} \dd x_1\dd x_2 \dd x_3 \int\dd W_{x_1,x_2}^{2\beta}(\omega_1) \dd W_{x_2,x_3}^{2\beta}(\omega_2) \\
\Bigl( 1 - \e{-\overline U(\omega_{1,1}, \omega_{2,1})} \Bigr) \sum_{k\geq1} \frac{(k+2) z^{k+2}}{(4\pi\beta k)^{d/2}} \e{-\frac{|x_1-x_3|^2}{4\beta k}}.
\end{multline}
Using Lemma \ref{lem diff B bridges} and bounding the exponentials by $1$, we get
\be\nn 
A_2 \leq \frac{2^{d/2} z^2}{(4\pi\beta)^{d}} \Bigl[ g_{d/2-1}(z) + 2 g_{d/2}(z) \Bigr] \int K(x,y) \dd x \dd y. \\
\ee
The terms where no consecutive legs interact are
\bm\nn
A_3 = \frac1{2|\Lambda|} \int_{\Lambda^4} \dd x_1\dd x_2 \dd x_3 \dd x_4 \\
\int\dd W_{x_1,x_2}^{2\beta}(\omega_1) \int\dd W_{x_3,x_4}^{2\beta}(\omega_2) \Bigl( 1 - \e{-\overline U(\omega_{1,1}, \omega_{2,1})} \Bigr) \\
\sum_{k_1, k_2 \geq 1} \frac{(k_1 + k_2 + 2) z^{k_1 + k_2 + 2}}{(4\pi\beta k_1)^{d/2} (4\pi\beta k_2)^{d/2}} \e{-\frac{|x_2-x_3|^2}{4\beta k_1} - \frac{|x_1-x_4|^2}{4\beta k_2}}.
\end{multline}
Then
\be \nn 
\begin{split}
A_3 &\leq \frac{2^{d/2-1}}{(4\pi\beta)^{3d/2}} \int_{\bbR^{3d}} \dd x \dd y \dd z \e{-\frac{|x-y-2z|^2}{8\beta}} K(x,y) \\ 
&\quad\quad\quad \times \sum_{k_1, k_2 \geq 1} \frac{(k_1 + k_2 + 2) z^{k_1 + k_2 + 2}}{(k_1 k_2)^{d/2}} \\
&= \frac{z^2}{(4\pi\beta)^d} g_{\frac d2}(z) \bigl[ g_{\frac d2-1}(z) + g_{\frac d2}(z) \bigr] \int K(x,y) \dd x \dd y.
\end{split}
\ee

We now decompose the terms in $B$ as $B_1+B_2+B_3$ according to the winding numbers of $\omega_1$ and $\omega_2$. The term $B_1$ involves two paths of winding numbers 1, and with the aid of Lemma \ref{lem diff B bridges} we find
\be\nn 
B_1 \leq \frac{z^2}{(2\pi\beta)^{d/2}} \int_{\bbR^d} K(x,x) \dd x.
\ee
Next, $B_2$ involves a path of winding number 1 and another path of higher winding number. Dropping the self-interaction terms yields the upper bound
\be\nn 
\begin{split}
B_2 &\leq \frac{2^{d/2}}{(4\pi\beta)^d} \int_{\bbR^{2d}} \dd x \dd y \,\e{-\frac{|x-y|^2}{8\beta}} K(x,y) \\ & \qquad \times\sum_{k\geq1} \frac{(k+2) z^{k+2}}{k^{d/2}} \\
&\leq \frac{2^{d/2} z^2}{(4\pi\beta)^d} [g_{d/2-1}(z) + 2 g_{d/2}(z)] \int K(x,y) \dd x \dd y.
\end{split}
\ee
Finally, $B_3$ involves paths with winding numbers higher than 2. We have 
\be\nn 
\begin{split}
B_3 &\leq \frac{2^{d/2}}{(4\pi\beta)^{\frac{3d}2}} \int_{\bbR^{3d}} \dd x \dd y \dd z \e{-\frac{|x-y-2z|^2}{8\beta}} K(x,y) \\
&\quad\quad\quad  \times \sum_{k_1, k_2 \geq 1} \frac{(k_1+1) z^{k_1+k_2+2}}{(k_1 k_2)^{d/2}} \\
&= \frac{z^2}{(4\pi\beta)^d} g_{\frac d2}(z) \bigl[ g_{\frac d2-1}(z) + g_{\frac d2}(z) \bigr] \int K(x,y) \dd x \dd y.
\end{split}
\ee
Collecting the bounds on $A_1,A_2,A_3,B_1,B_2,B_3$, we get the lower bound of Proposition \ref{prop technical bounds}.
\end{proof}

\section{Scattering estimates}
\label{sec scattering}

As before, let $U(x) \geq 0$ be radial and supported on the set $\{
x\, : \, |x|\leq R_0\}$. Let $a$ be the scattering length of $U$. We
consider the Hilbert space $L^2(\bbR^d)$ and the integral kernel
$K(x,y)$ of the operator $\e{2\beta \Delta} - \e{\beta(2\Delta - U)}$. It follows from the Feynman-Kac representation that
$K(x,y) \geq 0$, see Eq.\ \eqref{FK kernel} in the previous section.

We introduce 
\be\label{def:ab} 
a(\beta) = \frac 1 {8\pi \beta} \int
K(x,y) \dd x \dd y\,.  
\ee 
We shall see below that, for $d=3$,
$a(\beta)$ is a good approximation to the scattering length. In fact,
$a\leq a(\beta)\leq a_0$, with $a_0$ the first order Born
approximation to $a$. In two dimensions $a(\beta)$ is dimensionless
and its relation to the scattering length is $a(\beta)\approx
|\ln(a^2/\beta)|^{-1}$ for large $\beta$.
For $t>0$, we also introduce the function
\[
f(t) = t \frac{ 1- \e{-t}}{t- 1+\e{-t}}.
\]

\begin{lemma}\label{varpr}
We have
\be\label{vpab}
a(\beta) = \frac 1{8\pi} \inf_{\psi \in H^1(\bbR^d)} \mathcal{E}_\beta(\psi),
\ee
where
\begin{align}\nn
\mathcal{E}_\beta(\psi) & = \int_{\bbR^d} \left(2 |\nabla \psi(x)|^2 + U(x) |1-\psi(x)|^2 \right) \dd x \\ \nn & \quad 
+ \frac 1\beta \left\langle \psi\left| f(\beta(-2\Delta+U)) \right|\psi\right\rangle.
\end{align}
\end{lemma}

Note that $f$ is monotone decreasing, with $1\leq f(t)\leq 2$ for all
$t>0$. From monotonicity it follows immediately that $a(\beta)$ is
monotone decreasing in $\beta$. Moreover, for $d=3$ it is not hard to
see that $\lim_{\beta \to \infty} a(\beta) = a$. For any $d$,
$\lim_{\beta\to 0} a(\beta) = (8\pi)^{-1} \int U(x) \dd x$.  (This
is also true when $\int U(x)\dd x =\infty$.)

\begin{proof}
We first consider the case when $U$ is bounded. With the aid of the Duhamel formula we have
\begin{align}\nn
&\e{2\beta \Delta} - \e{\beta(2\Delta -  U)}  = \int_0^\beta \e{2 (\beta-t) \Delta} U \e{t (2\Delta -  U)} \dd t \\ \nn
&=  \int_0^\beta \e{2 (\beta-t) \Delta} U \e{2t \Delta } \dd t \\ &
\quad - \int_0^\beta \int_0^t  \e{2 (\beta-t) \Delta} U \e{s (2\Delta- U)} U \e{2(t-s) \Delta} \dd s \dd t\,. \nn 
\end{align}
Hence
\be\label{eab}
a(\beta) = \frac 1{8\pi} \int U(x) \left ( 1  - \psi_\beta(x)\right) \dd x,
\ee
where $\psi_\beta(x) = (L_\beta U)(x)$, with 
\be \nn
L_\beta =  \int_0^\beta (1-s/\beta)   \e{s(2\Delta- U)} \dd s.
\ee

The functional $\mathcal{E}_\beta(\psi)$ has a quadratic and a linear part in $\psi$, and it is not hard to see that the unique minimizer satisfies
\be\label{mt}
\big[ -2\Delta + U + \tfrac 1\beta f(\beta(-2\Delta+U)) \big] \psi = U .
\ee
Since 
$$
\frac {1}{t + f(t)} = \int_0^1 (1-s) \e{-st} \dd s
$$
it follows that $\psi=\psi_\beta$, i.e.\ $\psi_\beta = L_\beta U$ is the unique minimizer of $\mathcal{E}_\beta$. 
After multiplying (\ref{mt}) by $\psi_\beta$ and integrating we see that $\mathcal{E}_\beta(\psi_\beta) = \int U(x) ( 1-\psi_\beta(x))\dd x$ which, because of (\ref{eab}),  implies  (\ref{vpab}).

Finally, the case of unbounded $U$ can be dealt with using monotone
convergence. If we replace $U$ by $U_s(x)=\min\{U(x),s\}$ then the
kernel $K(x,y)$ corresponding to $U_s$ is monotone increasing in
$s$. We can apply the argument above to $U_s$ and take the limit $s\to
\infty$ at the end. The convergence of $a(\beta)$ is guaranteed by
monotonicity.
\end{proof}

The variational principle of Lemma~\ref{varpr} is convenient for obtaining an upper bound on $a(\beta)$.

\begin{lemma}\label{lem scattering 1}
For $d=2$ and  $\sqrt\beta\, |\ln(a/\sqrt\beta)|^{-1} > R_0$, 
\be \label{aab}
a(\beta) \leq \frac 1{ |\ln(a^2/\beta)| - 2 \ln |\ln (a/\sqrt\beta)|}  +\frac 1{ |\ln(a^2/\beta)|^{2}} \,.
\ee
For $d\geq 3$ and $a \sqrt\beta > R_0^{d-1}$, 
\bm\label{aabb}
a(\beta)\leq  \frac{\pi^{d/2-1}}{2\, \Gamma(d/2)} a \Bigl( \left[ 1- \big ( a \beta^{1-d/2}\big)^{1/(d-1)}\right]^{-1} \\
+ \frac 1d \big ( a \beta^{1-d/2}\big)^{1/(d-1)}\Bigr)\,.
\end{multline}
\end{lemma}
Note that the prefactor in (\ref{aabb}) is equal to $1$ for
$d=3$. Lemma~\ref{lem scattering 1} is the only place where the
finiteness of the range $R_0$ of $U$ is being used. Appropriate upper
bounds on $a(\beta)$ can also be obtained without this assumption,
and hence our main results generalize to repulsive interaction potentials with
infinite range (but finite scattering length). For simplicity, we
shall not pursue this generalization here.

\begin{proof}
Let $R>R_0$, and let $\psi_\infty$ be the minimizer of 
\be\label{infly}
\int_{|x|\leq R} \left( 2 |\nabla\psi|^2 + U |1-\psi|^2 \right) \dd x
\ee
subject to the boundary condition $\psi(x)=0$ for $|x|=R$. It can be shown \cite{LY2d,LSSY} that there exists a unique minimizer for this problem, which satisfies $0\leq \psi_\infty\leq 1$
 and 
$$
\psi_\infty(x) = \begin{cases} 1 - \frac{\ln(|x|/a)}{\ln(R/a)} & \text{for $d=2$} \\ 1- \frac{ 1- a |x|^{2-d}} {1-a R^{2-d}} & \text{for $d\geq 3$} \end{cases}
$$
in the region $R_0\leq |x|\leq R$. Moreover, 
the minimum of (\ref{infly}) is given by 
$$
E_R = \begin{cases} \frac{4\pi}{\ln(R/a)} & \text{for $d=2$} \\ \frac{4\pi^{d/2} a}{\Gamma(d/2) (1-a R^{2-d})} & \text{for $d\geq 3$.} \end{cases}
$$

To obtain an upper bound on $a(\beta)$, we use the variational principle (\ref{vpab}) with $\psi(x) = \psi_\infty(x)$ for $|x|\leq R$, and $\psi(x)=0$ for $|x|\geq R$. Using $|\psi_\infty|\leq 1$ and $f\leq 2$, we obtain the bound
$$
a(\beta) \leq \frac{E_R}{8\pi} + \frac{\sigma_d R^d}{4\pi\beta}\,,
$$
where $\sigma_d = \pi^{d/2} /\Gamma(1+d/2)$ denotes the volume of the unit ball in $\bbR^d$. The choice $R=\sqrt{\beta}\,[\ln(\sqrt\beta/a)]^{-1}$ for $d=2$ and $R=(a\sqrt\beta)^{1/(d-1)}$ for $d\geq3$ yields (\ref{aab}) and (\ref{aabb}).
\end{proof}

For our lower bound on the density in Proposition~\ref{prop technical bounds}, we need a bound on two more integrals of the kernel $K$. Since they appear only in terms of higher order, a rough bound will do.

\begin{lemma}
\label{lem scattering 2}
Let 
\be\label{def:apb}
\begin{split}
a'(\beta) &= (8\pi \beta)^{d/2-1} \int K(x,x) \dd x,  \\
a''(\beta) &= (8\pi \beta)^{d/2-1} \int K(x,-x) \dd x \,. 
\end{split}
\ee
Then 
\be\nn
\max\{ a'(\beta), a''(\beta) \}  \leq 2^{d/2} a(\beta/2)\,.
\ee
\end{lemma}

For $d= 3$, it can be shown that both $a'(\beta)$ and $a''(\beta)$ converge to $a$ as $\beta\to \infty$, but we do not need this here.

\begin{proof}
Using the semi-group property of the heat kernel we can write
\bm\nn
K(x,z) = \int_{\bbR^d} \Bigl( \e{\beta\Delta}(x,y) \e{\beta\Delta}(y,z) \\ - \e{\beta(\Delta-\frac 12 U)}(x,y) 
\e{\beta(\Delta-\frac 12 U)}(y,z)\Bigr) \dd y\,.
\end{multline}
Since $a b -c d \leq a(b-d)+b(a-c)$ for $a\geq c$ and $b\geq d$, $K(x,z)$ is bounded above by 
\bm\nn
 \int_{\bbR^d} \e{\beta\Delta}(x,y)\left( \e{\beta\Delta}(y,z) - \e{\beta(\Delta-\frac 12 U)}(y,z)\right) \dd y \\
+ \int_{\bbR^d} \e{\beta\Delta}(y,z) \left( \e{\beta\Delta}(x,y) -\e{\beta(\Delta-\frac 12U)}(x,y) \right) \dd y\,. 
\end{multline}
Using the bound $\e{\beta\Delta}(x,y)\leq (4\pi\beta)^{-d/2}$ the claim follows easily.
\end{proof}

\section{Conclusion}
We have given rigorous upper bounds on the critical temperature for
two- and three-dimensional Bose gases with repulsive two-body
interactions. In two dimensions, our bound agrees to leading order in
$a^2\rho$ with the expected critical temperature for superfluidity. In
three dimensions, our bound shows that the critical temperature is not
greater than the one for the ideal gas plus a constant times $\sqrt
{a\rho^{1/3}}$.

Our bounds are based on the observation that the one-particle reduced
density matrix decays exponentially if the fugacity $z$ satisfies
$z<1$. What is needed are lower bounds on the particle density in the
grand canonical ensemble.  The Feynman-Kac path integral
representation allows us to get bounds in terms of certain integral
kernels which, in turn, can be estimated by the scattering length of the
interaction potential using a suitable variational principle.

\medskip {\footnotesize {\bf Acknowledgments:} It is a pleasure to thank Rupert Frank
  and Elliott Lieb for many stimulating discussions, and Markus Holzmann for helpful comments on the physics literature. D.U. is grateful
  for the hospitality of ETH Z\"urich, the Center of Theoretical
  Studies of Prague, and the University of Arizona, where parts of
  this project were carried forward. Partial support by the US
  National Science Foundation grants PHY-0652356 (R.S) and
  DMS-0601075 (D.U.) is gratefully acknowledged.  }

\end{document}